\documentclass[11pt,letterpaper]{article}
\usepackage{fullpage}
\usepackage{amsfonts}
\usepackage{graphicx}
\usepackage{amsmath}
\usepackage{amsthm}
\DeclareMathAlphabet{\mathcal}{OMS}{cmsy}{m}{n}
\DeclareMathAlphabet{\mathpzc}{OT1}{pzc}{m}{it}
\newcommand{\helper}{\mathcal{H}}
\newcommand{\verifier}{\mathcal{V}}
\newcommand{\fail}{\perp}
\newcommand{\algorithm}{\mathcal{M}}

\DeclareMathOperator{\hcost}{hcost}
\DeclareMathOperator{\vcost}{vcost}

\DeclareMathOperator{\occ}{occ}
\newcommand{\finger}{\mathfrak f}
\newcommand{\eat}[1]{}
\newcommand{\push}{\mathsf{push}}
\newcommand{\pop}{\mathsf{pop}}
\newcommand{\event}{\mathsf{event}}
\newcommand{\edge}{\mathsf{edge}}
\newcommand{\para}[1]{\medskip \noindent{\bf #1}}
\newtheorem{theorem}{\textbf{Theorem}}[section]
\newtheorem{definition}[theorem]{\textbf{Definition}}
\newtheorem{corollary}[theorem]{\textbf{Corollary}}
\newtheorem{lemma}[theorem]{\textbf{Lemma}}

\begin{document}
\title{Streaming Graph Computations with a Helpful Advisor}
\author{Graham Cormode\thanks{AT \& T Labs -- Research,
graham@research.att.com} 
 \and 
Michael Mitzenmacher\thanks{Harvard University, 
School of Engineering and Applied Sciences, michaelm@eecs.harvard.edu. This work was supported in part by NSF grants CCF-0915922 and CNS-0721491, and in part by grants from Yahoo! Research, Google, and Cisco, Inc.}
\and Justin Thaler\thanks{Harvard University, 
School of Engineering and Applied Sciences,
jthaler@seas.harvard.edu. Supported by the Department of Defense (DoD) through the National Defense Science \& Engineering Graduate Fellowship (NDSEG) Program.}}
\maketitle
\begin{abstract}
Motivated by the trend to outsource work to commercial cloud computing
services, we consider a variation of the streaming paradigm where a
streaming algorithm can be assisted by a powerful helper that can
provide annotations to the data stream.  We extend previous work on
such {\em annotation models} by considering a number of graph
streaming problems.  Without annotations, streaming algorithms for 
graph problems generally require significant memory;  we show that
for many standard problems, including all graph problems that can 
be expressed with totally unimodular integer programming formulations,
only a constant number of hash values are 
needed for single-pass algorithms given linear-sized annotations.
We also obtain a protocol achieving \textit{optimal} tradeoffs between
annotation length and memory usage for matrix-vector multiplication;
this result contributes to a trend of recent research on numerical
linear algebra in streaming models.
\end{abstract}

\section{Introduction} 
\label{intro}
The recent explosion in the number and scale of real-world structured
data sets including the web, social networks, and other relational
data has created a pressing need to efficiently process and analyze
massive graphs. This has sparked the study of graph algorithms that
meet the constraints of the standard streaming model: restricted
memory and the ability to make only one pass (or few passes) over
adversarially ordered data.  However, many results for graph streams
have been negative, as many foundational problems require either 
substantial working memory or a prohibitive number of passes over the data
\cite{graphmining08}.  Apparently most graph algorithms fundamentally
require flexibility in the way they query edges, and therefore the
combination of adversarial order and limited memory makes many
problems intractable.

To circumvent these negative results, variants and
relaxations of the standard graph streaming model have been proposed, 
including the Semi-Streaming 
\cite{semistream2}, W-Stream \cite{wstream}, Sort-Stream
\cite{sortstream}, Random-Order \cite{graphmining08}, and Best-Order
\cite{bestorder} models. 
In Semi-Streaming, memory requirements are relaxed,
allowing space proportional to the number of vertices in
the stream but not the number of edges.  
The W-Stream model allows the
algorithm to write temporary streams to aid in computation. 
And, as their names suggest, the Sort-Stream, Random-Order, and Best-Order
models relax the assumption of adversarially ordered input.  
The Best-Order model, for example, allows the input stream to be re-ordered
arbitrarily to minimize the space required for the computation.

In this paper, our starting point is a relaxation of the
standard model, closest to that put forth by Chakrabarti {\em et al.}
\cite{icalp09}, called the \textit{annotation model}. 
Motivated by recent work on outsourcing of database
processing,
as well as commercial cloud computing services such as
Amazon EC2, the annotation model allows access to a powerful advisor,
or \textit{helper} who observes the stream concurrently with
the algorithm. 
Importantly, in many of our motivating applications,
the helper is not a trusted entity: the commercial stream processing
service may have executed a buggy algorithm, experienced a hardware
fault or communication error, or may even be deliberately deceptive
\cite{bestorder,icalp09}. As a result, we require our protocols to be
\textit{sound}: our {\em verifier} must detect any lies or deviations
from the prescribed protocol with high probability.

The most general form of the annotation model allows the helper
to provide additional annotations in the data stream {\em at any point} to
assist the verifier, and one of the cost measures is the total length
of the annotation.  
In this paper, however, we focus on the case where the helper's
annotation arrives as a single message after both the helper
and verifier have seen the stream.  
The helper's message is also processed as a stream, since it may be
large; it often (but not always) includes a
re-ordering of the stream into a convenient form, as well as
additional information to guide the verifier.  
This is therefore stronger than the Best-Order model, which only
allows the input to be reordered and no more; but it is weaker than
the more general {\em online} model,  because in our model 
the annotation appears only after the input stream has finished. 

We argue that this model is of interest for several reasons. 
First, it requires minimal coordination between helper and verifier, 
since it is not necessary to ensure that annotation and stream data
are synchronized. 
Second, it captures the case when the verifier uploads data to 
the cloud as it is collected, 
and later poses questions over the data to the helper. 
Under this paradigm, the annotation must come \textit{after} the
stream is observed.
Third, we know of no non-trivial problems which separate the
general online and our ``at-the-end" versions of the model, 
and most prior results are effectively in this model. 

%
\eat{
Our algorithms therefore utilize a model that is stronger than the Best-Order model,
in that it allows for additional annotation information; however, it
is weaker than the online annotation model, in that we utilize only a
single terminal message (and always, in our work, provide a
re-ordering of the stream data).  We refer to this as the {\em weak
annotation model}.  Given that we are working with graph data,
re-ordering appears necessary to obtain useful results; however, we
leave as open questions obtaining better results by using either
weaker or stronger models.
}

Besides being practically motivated by  outsourced
computations, annotation models are closely related to Merlin-Arthur
proofs with space-bounded verifiers, and studying what can (and
cannot) be accomplished in these models is of independent
interest.

\para{Relationship to Other Work.}
Annotation models were first explicitly studied by
Chakrabarti {\em et al.} in
\cite{icalp09}, and focused primarily on protocols for 
canonical problems in numerical streams, such as Selection,
Frequency Moments, and Frequent Items. 
The authors also provided protocols for some graph problems: 
counting triangles, connectedness, and bipartite perfect matching.  
The Best-Order Stream Model was put forth by Das Sarma et
al. in \cite{bestorder}. 
They present protocols requiring logarithmic
or polylogarithmic space (in bits) for several problems, including perfect
matching and connectivity.   
Historical antecedents for this work are due to Lipton \cite{Lipton-old}, who
used fingerprinting methods to verify polynomial-time
computations in logarithmic space. 
Recent work  verifies shortest-path  computations
using cryptographic primitives, using polynomial space for the
verifier~\cite{ICDE10}. 

\para{Our Contributions.}
We identify two qualitatively different approaches to 
producing protocols for problems on graphs with $n$ nodes and
$m$ edges. 
In the first, the helper directly proves matching 
upper and lower bounds on a quantity. 
Usually, proving one of the two bounds is trivial: the helper
provides a feasible solution to the problem. 
But proving {\em optimality} of the feasible solution can be more
difficult, requiring the use of structural properties of the problem. 
In the second, we simulate the execution of a non-streaming
algorithm, using the helper to maintain the algorithm's internal data
structures to control the amount of memory used by the
verifier. 
The helper must provide the contents of the
data structures so as  to limit the amount of
annotation required.

Using the first approach (Section \ref{sec:direct}), we show that only
constant space
and annotation linear in the input size $m$ is needed to
determine whether a directed graph is a DAG and to compute the
size of a maximum matching.  We describe this as an $(m, 1)$ protocol,
where the first entry refers to the annotation size (which we also call the $\hcost$) and the second to the memory
required for the verifier (which we also call the $\vcost)$.
Our maximum matching result significantly extends
the bipartite perfect matching protocol of \cite{icalp09}, and is
tight for dense graphs, in the sense that there is a lower
bound on the \textit{product} of $\hcost$ and $\vcost$ of $\hcost
\cdot \vcost = \Omega(n^2)$ bits for this problem.
Second, we define a streaming version of the
linear programming problem, and provide an $(m, 1)$ protocol.  
By exploiting duality, 
we hence obtain $(m,1)$ protocols for many graph problems with
totally unimodular integer programming formulations, including
shortest $s$-$t$ path, max-flow, min-cut, and minimum-weight bipartite
perfect matching. We also show all are tight by proving lower
bounds of $\hcost \cdot \vcost = \Omega(n^2)$ bits for all four problems.
A more involved protocol
obtains \textit{optimal} tradeoffs between annotation cost and working memory for 
dense LPs and matrix-vector multiplication;
this complements recent results on approximate linear algebra 
in streaming models (see e.g. \cite{linalg1,linalg2}). 

For the second approach (Section \ref{sec:transcript}), we make use of
the idea of ``memory checking'' due to Blum et
al. \cite{memcheck}, which allows a small-space verifier to outsource
data storage to an untrusted server. 
We present a general simulation theorem based on this checker, and obtain as
corollaries tight protocols for a variety of canonical graph problems.
In particular, we give an $(m, 1)$ protocol for verifying a minimum
spanning tree, an $(m + n \log n, 1)$ protocol for single-source
shortest paths, and an $(n^3, 1)$ protocol for all-pairs shortest
paths. 
We provide a lower bound of $\hcost \cdot \vcost = \Omega(n^2)$
bits for the latter two problems, and an identical lower bound for MST
when the edge weights can be given incrementally.  
While powerful, this technique has its limitations: there does not
seem to be any generic way to obtain the same kind of tradeoffs
observed above. 
Further, there are some instances where direct application of memory
checking does not achieve the best bounds for a problem. 
We demonstrate this by presenting an $(n^2 \log n,
1)$ protocol to find the diameter of a graph; 
this protocol leverages the ability to use randomized
methods to check computations more efficiently than via generating or
checking a deterministic witness. 
In this case, we rely on techniques to verify 
matrix-multiplication
 in quadratic time, and show that this is tight via   
a nearly matching
lower bound for diameter of $\hcost \cdot \vcost = \Omega(n^2)$.

In contrast to problems on numerical streams, where it is often trivial
to obtain $(m, 1)$ protocols by replaying the stream in sorted order, 
it transpires that achieving linear-sized
annotations with logarithmic space
 is more challenging for many graph problems.
Simply providing the solution (e.g. a graph matching or spanning tree)
is insufficient, since we have the additional burden of demonstrating
that this solution is {\em optimal}. 
A consequence is that we are able to provide solutions to several
problems for which no solution is known in the best-order model
(even though one can reorder the stream in the best-order model 
so that the ``solution'' edges arrive first). 


\section{Model and Definitions}
Consider a data stream $\mathcal{A} = \langle a_1 , a_2 , \dots, a_m
\rangle$ with each $a_i$ in some universe $\mathcal{U}$. 
Consider a
probabilistic {\em verifier} $\verifier$ who observes $\mathcal{A}$ and a
deterministic {\em helper} $\helper$ who also observes $\mathcal{A}$ and can
send a message $\mathpzc{h}$ to $\verifier$ after $\mathcal{A}$ has
been observed by both parties.  
This message, also referred to as an
{\em annotation}, should itself be interpreted as a data stream that
is parsed by $\verifier$, which may permit $\verifier$ to use
space sublinear in the size of the annotation itself.  
That is, $\helper$ provides an annotation $\mathpzc{h}(\mathcal{A}) = (
\mathpzc{h}_1(\mathcal{A}), \mathpzc{h}_2(\mathcal{A}), \dots
\mathpzc{h}_\ell(\mathcal{A}) )$.  

We study randomized streaming protocols for computing functions
$f(\mathcal{A}) \rightarrow \mathbb{Z}$. 
Specifically, assume $\verifier$ has access to a private
random string $\mathcal{R}$ and at most $w(m)$ machine words of working memory,
and that $\verifier$ has one-way access to the input 
$\mathcal{A} \cdot \mathpzc{h}$,
where $\cdot$ represents  concatenation. Denote the output of
protocol $\mathcal{P}$ on input $\mathcal{A}$, given helper
$\mathpzc{h}$ and random string $\mathcal{R}$, by $out(\mathcal{P},
\mathcal{A}, \mathcal{R}, \mathpzc{h})$.  We allow $\verifier$ to
output $\perp$ if $\verifier$ is not convinced that the annotation is
valid. We say that $\mathpzc{h}$ is \textit{valid} for
$\mathcal{A}$ with respect to $\mathcal{P}$ if
$Pr_{\mathcal{R}}(out(\mathcal{P}, \mathcal{A}, \mathcal{R},
\mathpzc{h}) = f(\mathcal{A})) =1 $, and we say that
$\mathpzc{h}$ is $\delta$-\textit{invalid} for $\mathcal{A}$ with
respect to $\mathcal{P}$ if $Pr_{\mathcal{R}}(out(\mathcal{P},
\mathcal{A}, \mathcal{R}, \mathpzc{h}) \neq \perp) \leq \delta$.  We
say that $\mathpzc{h}$ is a \textit{valid helper} if
$\mathpzc{h}$ is valid for all $\mathcal{A}$.  We
say that $\mathcal{P}$ is a valid protocol for $f$ if 

\begin{enumerate}
\item
There exists at least one valid helper $\mathpzc{h}$ 
with respect to $\mathcal{P}$ and 

\item
For all helpers $\mathpzc{h}' $ and all streams $\mathcal{A}$,
either $\mathpzc{h}'$ is valid for $\mathcal{A}$ or 
$\mathpzc{h}'$ is $\frac{1}{3}$-invalid for $\mathcal{A}$.
\end{enumerate}

\noindent
Conceptually, $\mathcal{P}$ is a valid protocol for $f$ if for each
stream $\mathcal{A}$ there is \textit{at least} one way to convince
$\verifier$ of the true value of $f(\mathcal{A})$, 
and $\verifier$ rejects all other
annotations as invalid (this differs slightly from \cite{icalp09} 
to allow for multiple $\mathpzc{h}$'s that can convince $\verifier$).
The constant $\frac{1}{3}$ can be any constant less than $\frac{1}{2}$. 


Let $\mathpzc{h}$ be a valid helper chosen to minimize
the length of $\mathpzc{h}(\mathcal{A})$ for all $\mathcal{A}$. 
We define the help cost $\hcost(\mathcal{P})$ to be the maximum length of
$\mathpzc{h}$ over all $\mathcal{A}$ of length $m$, and the
verification cost $\vcost(P)=w(m)$, the amount of working memory used by
the protocol $P$.  
All costs are expressed in machine words of size $\Theta(\log m)$ bits, i.e. we assume 
any quantity polynomial in the input size can be stored in a constant
number of words; in contrast, lower bounds are expressed in bits. 
We say that $\mathcal{P}$ is an $(h, v)$ protocol
for $f$ if $\mathcal{P}$ is valid and $\hcost(\mathcal{A} ) = O(h +
1)$, $\vcost(\mathcal{A} ) = O(v + 1)$.  
While both $\hcost$ and $\vcost$ are natural costs for such protocols, 
we often aim to achieve
a $\vcost$ of $O(1)$ and then minimize $\hcost$. 
In other cases, we show that $\hcost$ can be decreased by
increasing $\vcost$, and study the tradeoff between these two
quantities. 

In some cases, $f$ is not  
a function of $\mathcal{A}$ alone; instead it depends on
$\mathcal{A}$ \textit{and} $h$. 
In such cases, 
$\verifier$ should simply \textit{accept}
if convinced that the annotation has the correct properties, 
and output $\fail$ otherwise. 
We use the same terminology as before, and
say that $\mathcal{P}$ is a valid protocol if
there is a valid helper and any $\mathpzc{h}'$ that is not valid for
$\mathcal{A}$ is $\frac{1}{3}$-invalid for $\mathcal{A}$.

In this paper we primarily consider graph streams, which are streams
whose elements are edges of a graph $G$. 
More formally, consider a
stream $\mathcal{A} = \langle e_1 , e_2 , \dots, e_m \rangle$ with
each $e_i \in [n] \times [n]$.  
Such a stream defines a (multi)graph
$G = (V , E)$ where $V = \{v_1 , . . . , v_n \}$ and $E$ is the
(multi)set of edges that naturally corresponds to $\mathcal{A}$.  
We use the notation $\{i: m(i)\}$ for the multiset in which $i$ appears
with multiplicity $m(i)$.  
Finally, we will sometimes consider graph
streams with directed edges, and sometimes with weighted edges; in the
latter case each edge $e_i \in [n] \times [n] \times \mathbb{Z}_+$.
\subsection{Fingerprints}
Our protocols make careful use of 
\textit{fingerprints}, permutation-invariant hashes 
that can be efficiently computed in a streaming fashion. 
They determine in small space (with high probability) 
whether two streams have identical frequency distributions. 
They 
are the workhorse of algorithms proposed in earlier work on 
streaming models with an untrusted helper
 \cite{bestorder,icalp09,Lipton-old,Lipton-older}.
We
sometimes also need the fingerprint function to be linear. 
\begin{definition}[Fingerprints]
A fingerprint of a multiset $M = \{i : m(i)\}$ where each $i \in [q]$
  for some known upper bound $q$
is defined as a
  computation over the finite field with $p$ elements, $\mathbb{F}_p$,
as $\finger_{p,\alpha}(M) = \sum_{i=1}^q m(i) \alpha^i$, where $\alpha$ is chosen
uniformly at random from $\mathbb{F}_p$. 
We typically leave $p, \alpha$ implicit, and just write $\finger(M)$. 
\end{definition}
Some properties of $\finger$ are immediate:
it is linear in  $M$, and can easily be computed incrementally as
elements of $[q]$ are observed in a stream one by one. 
The main property of $\finger$ is that 
$\Pr[ \finger(M) = \finger(M') | M \neq M'] \leq q/p$ over the random
choice of $\alpha$ (due to  standard properties of
polynomials over a field).
Therefore, if  $p$ is sufficiently large, say, polynomial in $q$
and in an (assumed) upper bound on the multiplicities $m(i)$, then
this event  happens with only polynomially small probability. 
For cases when the domain of the
multisets is not $[q]$, we either establish a bijection to
$[q]$ for an appropriate value of $q$, or use a hash
function to map the domain onto a large enough $[q]$ such that
there are no collisions with high probability (whp).  
In all cases, $p$ is chosen to be $O(1)$ words. 

A common subroutine of many of our protocols forces $\helper$
to provide a ``label'' $l(u)$ for each node upfront, and then replay
the edges in $E$, 
with each edge $(u, v)$ annotated with $l(u)$ and $l(v)$
so that
 each instance of each node $v$ appears with the {\em same} label $l(v)$. 
\begin{definition} \label{labeldef} 
We say a list of edges $E'$ is {\em label-augmented} if 
(a) $E'$ is preceded by a sorted list of all the nodes $v\in V$, each 
with a value $l(v)$ and $\deg(v)$, where $l(v)$ is the label of $v$
and $\deg(v)$ is claimed to be the degree of $v$; and
(b) each edge $e=(u, v)$ in $E'$ is annotated with a pair of
symbols $l(e,u)$ and $l(e,v)$.  
\noindent
We say a list of label-augmented edges $E'$ is valid if 
for all edges $e=(u,v)$, $l(e,u)=l(u)$ and $l(e,v)=l(v)$; and
$E'=E$, where $E$ is the set of edges observed in the stream $A$.
\end{definition}
\begin{lemma}[Consistent Labels] \label{labellemma}
There is a valid $(m, 1)$ protocol 
that accepts any valid list of label-augmented edges.
\end{lemma}
\begin{proof} 
$\verifier$ uses the annotation from  Definition \ref{labeldef}~(a) to 
make a fingerprint of the multiset $S_1 :=\{(u, l(u)): \deg(u)\}$. 
$\verifier$ also maintains a fingerprint $f_1$ of all 
$(u, l(e,u))$ pairs seen while observing the edges of $L$.
If $f_1 = \finger(S_1)$ then (whp) each node $u$
must be presented with label $l(e,u)=l(u)$ 
every time it is reported in an edge $e$
(and moreover $u$ must be reported in exactly $\deg(u)$ edges),
else the multiset of observed (node, label) pairs 
would not match  $S_1$.  
Finally, $\verifier$ ensures that $E'=E$
by checking that $\finger(E) =\finger(E')$. 
\end{proof}
\section{Directly Proving Matching Upper and Lower Bounds}
\label{sec:direct}
\subsection{Warmup: Topological Ordering and DAGs}
A (directed) graph  $G$ is a DAG if and only if $G$ has a
\textit{topological ordering}, which is  
an ordering of $V$ as $v_1, \dots v_n$ such that for every
edge $(v_i, v_j)$ we have $i < j$ \cite[Section 3.6]{KandT}.
Hence, if $G$ is a DAG, $\helper$ can prove it by providing a
topological ordering. If $G$ is not a DAG, $\helper$ 
can provide a directed cycle as witness.
\begin{theorem} \label{DAGthrm} There is a valid $(m, 1)$ protocol to
  determine if a graph is a DAG. 
\end{theorem} 
\begin{proof} 
If $G$ is not a DAG, $\helper$ provides 
a directed cycle 
$C$ as $(v_1, v_2),$ $(v_2, v_3)$ $\ldots$ $(v_k, v_1)$.
To ensure  $C \subseteq E$, $\helper$ then provides $E \setminus
C$, allowing $\verifier$ to 
check that $\finger(C \cup (E \setminus C)) = \finger(E)$.

\noindent
If $G$ is a DAG, let  $v_1, \dots v_n$ be a topological ordering of $G$. 
We require $\helper$ to replay the edges of $G$, with 
edge $(v_i, v_j)$ annotated with the ranks of $v_i$ and $v_j$
i.e. $i$ and $j$. We ensure $\helper$ provides consistent ranks 
via the Consistent Labels protocol of Lemma \ref{labellemma}, with the
ranks as ``labels''. If any edge $(v_i, v_j)$ is presented with $j > i$,
$\verifier$  rejects immediately.
\end{proof}
\subsection{Maximum Matching}
We give an $(m, 1)$ protocol for maximum matching which
leverages the combinatorial structure of the problem. 
Previously, matching was only  studied in the bipartite 
case, where an $(m, 1)$ protocol and a lower bound of
$\hcost \cdot \vcost = \Omega(n^2)$ bits for dense graphs were shown 
\cite[Theorem 11]{icalp09}.
The same lower bound applies to the more general problem of maximum
matching, so our protocol is tight up to logarithmic factors. 

The protocol shows matching upper and lower bounds on the size of the
maximum matching. 
Any feasible matching presents a lower bound. 
For the upper bound 
we appeal to the Tutte-Berge formula \cite[Chapter 24]{textbook}: 
the size of a maximum matching of a graph $G = (V, E)$ is equal to
%
$ \frac{1}{2} \min_{V_S\subseteq V}
  (|V_S|-\occ(G-V_S)+|V|),$
  where $G-V_S$ is the subgraph of $G$ obtained by deleting the vertices
  of $V_S$ and all edges incident to them, and
  $\occ(G-V_S)$ is the number of components in the graph $G-V_S$
  that have an odd number of vertices.
So for any set of nodes $V_S$, $\frac{1}{2} (|V_S|-\occ(G-V_S)+|V|)$ is 
an upper bound on the size
of the maximum matching, and there exists some $V_S$ for
which this quantity equals the size of a maximum matching $M$. 
Conceptually, providing both $V_S$ and $M$, $\helper$ 
proves that the maximum matching size is $M$. 
Additionally, $\helper$ has to provide a proof of the value of 
$\occ(G-V_S)$ to $\verifier$. 

\begin{theorem} \label{maxmatching}
There is a valid $(m, 1)$ protocol for maximum matching. Moreover, any
protocol for max-matching requires $\hcost \cdot \vcost = \Omega(n^2)$ bits.
\end{theorem}
\begin{proof} To prove a lower bound of $k$ on the size of the maximum matching,
$\helper$ provides a matching $M = (V_M, E_M)$ of size $|E_M|=k$, 
and then proves that $M$ is indeed a matching. 
It suffices to prove that $|V_M| = 2|E_M|$ and 
$M \subseteq E$.
First, $\helper$ lists $E_M$, and $\verifier$
fingerprints the nodes present as $f(V_M)$. 
$\helper$ then presents $V'_M$ which is claimed to be
$V_M$ in sorted order, allowing
$\verifier$ to easily check no node appears more than once
and that $\finger(V_M) = \finger(V'_M)$. 
Next, $\helper$ provides $E \setminus M$, 
allowing $\verifier$ to check that
$\finger(M \cup (E \setminus M)) = \finger(E)$.
Hence $M$ is a matching.

To prove an upper bound of $k$ on the size of the maximum matching,
$\helper$ sends a (sorted) set $V_S \subseteq V$, 
where $\frac{1}{2}(|V_S|-\occ(G-V_S)+|V|) = k$. 
Both $|V_S|$ and $|V|$ are computed directly; for $\occ(G-V_S)$, 
 $\helper$ sends a sequence  
of (sub)graphs $C_i =(V_i, E_i) \subseteq V \times E$ claimed to be a partition of $G-V_S$ into
connected components.
$\verifier$ can easily compute $c$, the number of $C_i$'s with an 
odd number of nodes.
To ensure that the $C_i$'s are indeed the connected
components of $G-V_S$, it suffices to show that
(a) each $C_i$ is connected in $G-V_S$;
(b) $V \setminus V_S$ is the disjoint union of the $V_i$'s; 
and 
(c) there is no edge $(v,w) \in E$ s.t. $v \in V_i, w \in V_j, i \neq j$.

To prove Property (a), $\helper$ presents the (sub)graph $C_i$ as 
$V_i \subset V$ 
(in sorted order) where each $v$ is paired with its degree $\deg(v)$;
followed by $E_i \subset E$ (in arbitrary order). 
Fingerprints are used to ensure that the multiset of nodes present in
$E_i$ matches the claimed degrees of  nodes in $V_i$. 
If these fingerprints agree, then (whp) 
$E_i \subseteq V_i \times V_i$.
Then $\helper$ uses the connectivity protocol from \cite[Theorem 5.6]{icalp09}
on the (sub)graph $C_i = (V_i, E_i)$ to prove that $C_i$ is connected.  
 Each of these checks on $C_i$ has $\hcost$ $O(|E_i|)$.
Note that $\verifier$ requires only a constant number of fingerprints
for these checks, and can use the same working memory for each
different $C_i$ to check that
$E_i \subseteq V_i \times V_i$ and that $C_i$ is connected. 
The total $\vcost$ over all $C_i$ is a constant number of fingerprints; 
the total $\hcost$ is $O(m)$. 


Property (b) is checked by testing
$\finger\big((\cup_i V_i) \cup V_S\big) = \finger(V)$, where the unions in the LHS
count multiplicities; if the fingerprints match then whp $V \setminus V_S$ is the
 \textit{disjoint} union of the $V_i$'s.
For (c), it suffices to ensure that each each edge in $E \setminus(\bigcup_i E_i)$ 
is incident to at least one node in $V_S$, as we have already checked that no edges
in $\bigcup_i E_i$ cross between $V_i$ and $V_j$ for $i \neq j$.
To this end, we use the ``Consistent Labels" protocol of Lemma \ref{labellemma}, with $l(u)=1$ indicating
$u \in V_S$ and $l(u)=0$ indicating $u \notin V_S$, to force $\helper$ to replay all of $E$ with each edge $(u, v)$ 
annotated with $l(u)$ and $l(v)$. This allows $\verifier$
to identify the set $E_S$ of edges incident to at least one node in $V_S$.
$\verifier$ checks that $\finger\big((\cup_i E_i) \cup E_S\big)=\finger(E)$, which ensures
(whp) that Property (c) holds and that over the entire partition of $G$ no edges are added or omitted.
Finally, provided all the prior fingerprint tests pass, the protocol
accepts if $c$, the number of $C_i$'s with an odd number of nodes, 
satisfies $\frac{1}{2}(|S|-c+|V|) = k$.
\end{proof}

%
\subsection{Linear Programming and TUM Integer Programs}
We present protocols to solve linear programming problems in our model
leveraging the theory of LP duality. 
This leads to  non-trivial schemes for a 
variety of graph problems.

\begin{definition} 
Given a data stream $\mathcal{A}$ containing 
entries of vectors $\mathbf{b} \in \mathbb{R}^b$, $\mathbf{c} 
\in \mathbb{R}^c$, and non-zero entries of a $b \times c$ matrix 
$A$ in some arbitrary order, possibly interleaved.
Each item in the stream indicates the index of the object
it pertains to.  
The LP streaming problem on $\mathcal{A}$ 
is to determine the value of the linear 
program  $\min \{\mathbf{c}^T \mathbf{x} \mbox{ } | \mbox{ }
A\mathbf{x} \leq \mathbf{b}$\}. 
\end{definition}

We present our protocol as if each entry of each
object appears at most once (if an entry does not appear, 
it is assumed to be zero). 
When this is not the case, 
the final value for that entry is interpreted as 
the {\em sum} of all corresponding values in the stream. 

\begin{theorem} \label{LPtheorem}
There is a valid $(|A|, 1)$ protocol for the LP streaming problem,
where $|A|$ 
is the number of non-zero entries in the constraint matrix $A$ of $\mathcal{A}$.
\end{theorem}

\begin{proof} 
The protocol shows an upper bound by providing a primal-feasible
solution $\mathbf{x}$, and a lower bound by providing a dual-feasible
solution $\mathbf{y}$. 
When the value of both solutions match, $\verifier$ is convinced that
the optimal value has been found. 

From the stream, $\verifier$ fingerprints the sets 
$S_A = \{(i, j, A_{i,j})\}$, 
$S_B = \{(i, \mathbf{b}_i)\}$ and 
$S_C = \{(i, \mathbf{c}_j)\}$. 
Then $\helper$ provides all pairs of values $\mathbf{c}_j, \mathbf{x}_j, 1 \leq j \leq c$,
with each  $\mathbf{x}_j$ additionally annotated with $|A_{\cdot j}|$,  
the number of non-zero entries in column $j$ of $A$. 
This allows $\verifier$ to fingerprint the multiset 
$S_X = \{(j, \mathbf{x}_j): |A_{\cdot j}|\}$
 and calculate 
 the solution cost $\sum_{j=1}^{b} \mathbf{c}_j \mathbf{x}_j$. 

To prove feasibility, for each row $i$ of $A$, $A_{i \cdot}$,
$\helper$ sends $\mathbf{b}_i$, then
(the non-zero entries of) $A_{i \cdot }$ so that $A_{ij}$ is annotated with
$\mathbf{x}_j$. 
This allows the $i$th constraint to be checked easily in constant
space. 
$\verifier$ fingerprints the values given by $\helper$ for 
$A$,  $\mathbf{b}$, and $\mathbf{c}$, and compares them to those for
the stream. 
A single fingerprint of the multiset of values presented for
$\mathbf{x}$ over all rows is compared to $\finger(S_X)$.
The protocol accepts $\mathbf{x}$ as feasible if all constraints are
met and all fingerprint tests pass. 

Correctness follows by observing that 
the agreement with $\finger(A)$ guarantees (whp) that each entry of
$A$ is presented correctly and no value is omitted.
Since $\helper$ presents
 each entry of $\mathbf{b}$ and $\mathbf{c}$ once, in index order, 
the fingerprints $\finger(S_B)$ and $\finger(S_C)$
ensure that these values are presented correctly. 
The claimed $|A_{\cdot j}|$ values must be correct: if not, then
the fingerprints of either $S_X$ or $S_A$ will not match the 
multisets provided by $\helper$. 
$\finger(S_X)$ also ensures that each time $\mathbf{x}_j$ is
presented, the same value is given (similar to Lemma \ref{labellemma}). 

To prove that $\mathbf{x}$ is primal-optimal, 
it suffices to show a feasible solution $\mathbf{y}$ 
to the dual $A^T$ so that $c^T \mathbf{x} = b^T \mathbf{y}$. 
Essentially we repeat the above protocol on the dual, 
and check that the claimed values are again consistent
with the fingerprints of $S_A, S_B, S_C$.
\end{proof}

For any graph problem that can be formulated as a linear program in 
which each entry of $A$, $\mathbf{b}$, and $\mathbf{c}$ can be derived
as a linear function of the nodes and edges, 
we may view each edge in a graph stream $\mathcal{A}$ as providing  
an update to values of one or more entries of $A$, $\mathbf{b}$, and
$\mathbf{c}$. 
Therefore, we immediately obtain a protocol for problems of this form
via Theorem \ref{LPtheorem}.  
More generally, we obtain protocols for problems formulated as
\textit{totally unimodular integer programs} (TUM IPs), since
optimality of a feasible solution is shown by
 a matching feasible solution of the dual of its LP relaxation
 \cite{LPbook}.  
\begin{corollary} \label{tumcorollary} 
There is a valid $(|A|, 1)$ protocol for any graph problem that can
be formulated  as a linear program or TUM IP in which each entry of
$A$, $\mathbf{b}$, and $\mathbf{c}$ is a linear function of the nodes
and edges of graph. 
\end{corollary}
This follows immediately from Theorem \ref{LPtheorem} and the
subsequent discussion: note that the linearity of the fingerprinting 
builds fingerprints of $S_A$, $S_B$ and $S_C$, so $\helper$ presents
only their (aggregated) values, not information from the unaggregated
graph stream. 
\begin{corollary}
\label{cor:lb}
Shortest $s-t$ path, max-flow, min-cut, and
minimum weight bipartite perfect matching (MWBPM) 
all have valid $(m, 1)$ protocols. 
For all four problems, a lower bound of $\hcost \cdot \vcost = \Omega(n^2)$ bits holds 
for dense graphs. 
\end{corollary}
\begin{proof} 
The upper bound follows from the previous corollary
because all the problems listed possess formulations as TUM IPs 
and moreover the constraint
matrix in each case has $O(m+n)$ non-zero entries.  
For example, for max-flow, $\mathbf{x}$ gives the flow on each edge, 
and the weight of each edge in the stream contributes (linearly) to constraints
on the capacity of that edge, and the flow through incident nodes. 

The lower bound for MWBPM, max-flow, and min-cut holds 
from \cite[Theorem 11]{icalp09} which argues
$\hcost \cdot \vcost = \Omega(n^2)$ bits for bipartite perfect matching, 
and straightforward reductions of bipartite perfect matching 
to all three problems, see e.g. \cite[Theorem 7.37]{KandT}.  
The lower bound for shortest $s-t$ path follows from a
 straightforward reduction from {\sc index}, for which
a lower bound linear in $\hcost \cdot \vcost$ was proven in \cite[Theorem 3.1]{icalp09}.
Given an instance $(x , k )$ of {\sc index} where $x \in \{0, 1\}^{n^2}$, 
$k \in [n^2]$, we construct graph $G$, with $V_G = [n+2]$, and 
$E_G = E_A \cup E_B$. 
Alice creates 
$E_A= \{(i, j): x_{f(i, j) =1}\}$ from $x$ alone, 
where $f$ is a 1-1 correspondence $[ n ] \times [ n ] \rightarrow [n^2]$.
Bob creates $E_B = \{(n+1, i), (j, n+2)\}$ using $f(i, j) = k$.
The shortest path between 
nodes $n+1$ and $n+2$ is 3 if $x_k = 1$ and is 4 or more
otherwise. 
This also implies that any approximation within $\sqrt{4/3}$
requires 
$\hcost \cdot \vcost = \Omega(n^2)$ (better inapproximability
constants may be possible). 
\end{proof}
\noindent
Conceptually, the above protocols for solving the LP streaming problem
are straightforward: $\helper$ provides a primal solution, potentially
repeating it once for each row of $A$ to prove feasibility, and
repeats the protocol for the dual. 
There are efficient protocols for
the problems listed in the corollary since the
constraint matrices of their IP formulations are sparse. For dense
constraint matrices, however, the bottleneck 
is proving feasibility. 
We observe that computing $A\mathbf{x}$ reduces to 
computing $b$ inner-product computations of vectors of dimension $c$. 
There are $(c^\alpha,c^{1-\alpha})$ protocols to verify such
inner-products \cite{icalp09}. 
But we can further improve on this
since one of
the vectors is held constant in each of the tests.
This reduces the
space needed by $\verifier$ to run these checks in parallel; moreover, we 
prove a lower bound of $\hcost \cdot \vcost = \Omega(\min(c, b)^2)$ bits, and 
so obtain an \textit{optimal} tradeoff for square
matrices, up to logarithmic factors.
%
\begin{theorem}
\label{thm:matrixvector}
Given a $b \times c$ matrix $A$ and a $c$ dimensional vector $\mathbf{x}$,
the product $A\mathbf{x}$ can be verified with a
valid $(bc^\alpha, c^{1-\alpha})$ protocol. Moreover, any such protocol
requires $\hcost \cdot \vcost = \Omega(\min(c, b)^2)$ bits for dense matrices.
\end{theorem}
\begin{proof}
We begin with the upper bound. The protocol for verifying inner-products which follows from
\cite{icalp09} treats a $c$ dimensional vector as an 
$h \times v$ array $F$, where $hv \geq c$. 
This then defines degree $c$ polynomials $f$ 
over a suitably large field, so that for each polynomial
$f(x,y) = F_{x,y}$. 
For an inner-product between two vectors, we wish to compute 
$\sum_{x \in [h], y \in [v]} F_{x,y} G_{x,y} = \sum_{x \in [h], y \in
  [v]} f(x,y) g(x,y)$ for the corresponding arrays $F, G$ and
polynomials $f, g$. 
These polynomials can then be evaluated at locations outside 
$[h] \times [v]$, so in the protocol 
$\verifier$ picks a random position $r$, and evaluates 
$f(r,y)$ and $g(r,y)$ for $1 \leq y \leq v$. 
$\helper$ then presents a degree $h$ polynomial $s(x)$ which is claimed to be 
$\sum_{y =1}^v f(x,y) g(x,y)$.
$\verifier$ checks that $s(r) = \sum_{y=1}^v f(r,y) g(r,y)$, and if so
accepts $\sum_{x=1}^{h} s(x)$ as the correct answer. 

In \cite{icalp09} it is shown how $\verifier$ can compute $f(r,y)$
efficiently as $F$ is defined incrementally in the stream: each
addition of $w$  to a particular index is mapped to 
$(x,y) \in [h] \times [v]$, which causes
$f(r,y) \gets f(r,y) + w p(r,x,y)$, where $p(r,x,y)$ depends
only on $x$, $y$, and $r$. 
Equivalently, the final value of $f(r,y)$ over updates in the stream 
where the $j$th update is $t_j = (w_j, x_j, y_j)$ is
$f(r,y) = \sum_{t_j: y_j = y} w_j p(r, x_j, y)$. 

To run this protocol over multiple vectors in parallel naively would
require keeping the $f(r,y)$ values implied by each different vector
separately, which would be costly.  
Our observation is that rather than keep these values explicitly, it
is sufficient to keep only a fingerprint of these values, and use the
linearity of fingerprint functions to finally test whether the
polynomials provided by $\helper$ for each vector together agree with
the stored values. 

In our setting, the $b \times c$ matrix $A$
implies $b$ polynomials of degree $c$. 
We evaluate each polynomial at $(r,y)$ for $1\leq y \leq v$ for the
same value of $r$: since each test is fooled by $\helper$ with small
probability, the chance that none of them is fooled can be kept high
by choosing the field to evaluate the polynomials over to have size
polynomial in $b+c$.
Thus, conceptually, the parallel invocation of $b$ instances of this
protocol require us to store 
$f_{i}(r,y)$ for $1 \leq y \leq v$ and $1 \leq i \leq b$ (for the $b$
rows of $A$), as well as $f_x(r,y)$ for $1 \leq y \leq v$ (where $f_x$ is the
polynomial derived from $\mathbf{x}$). 
Rather than store this set of $bv$ values explicitly, 
$\verifier$ instead stores only $v$ fingerprints, one for each
value of $y$, where each
fingerprint captures the set of $b$ values of $f_i(r,y)$.

From the definition of our fingerprints, this means over stream updates 
$t_j = (w_j, i_j, x_j, y_j)$ 
of weight $w_j$ to row $i_j$ and column indexed by $x_j$ and $y_j$
we compute a fingerprint
\[
\finger(A,y) = 
  \sum_{i=1}^{b} f_i(r,y) \alpha^i = 
 \sum_{i=1}^{b} \sum_{t_j: y_j=y, i_j=i} w_j p(r_j,x_j, y) \alpha^i 
\]
for each $y$, $1 \leq y \leq v$. Observe that for each $y$ this can be computed incrementally in the stream 
by storing only $r$ and the current value of $\finger(A,y)$. 

\eat{
In the stream, each update to $A_{ij}$ is interpreted by the protocol
of \cite{icalp09} as an update to a certain $f_i(r,y)$: this update is
to add a value which is a function of the update alone onto the
current value.  
So when we retain only a fingerprint, it is still possible to
propagate this update and track the fingerprint of the $f_i(r,y)$
values without needed explicit access to the prior values. 
More abstractly, the $f_i(r,y)$s can be expressed as a linear
transform of the input stream, and by the linearity of the
fingerprinting function, it is possible to compose these two linear
operations as a single linear operator. 
Hence this can be computed on the stream, only storing the
intermediate values of the fingerprint.  
}

To verify the correctness, $\verifier$ 
receives the $b$ polynomials $s_i$, and builds a fingerprint of 
the multiset of $S = \{i : s_i(r)\}$ incrementally. 
$\verifier$ then tests whether

\[
\sum_{y =1}^{v} \finger(A,y) f_x(r,y) = \finger(S) \]

To see the correctness of this, we expand the lhs, as
\begin{align*}
\sum_{y=1}^{v} \finger(A,y) f_x(r,y) = & 
 \sum_{y=1}^{v} \big(\sum_{i=1}^{b} f_i(r,y) \alpha^i\big) f_x(r,y)
\\
&=
\sum_{y=1}^{v} \big( \sum_{i=1}^{b} f_x(r,y) f_i(r,y) \alpha^i \big)
\\
&=
\sum_{i=1}^{b} \big( \sum_{y=1}^{v} f_x(r,y) f_i(r,y) \big) \alpha^i 
\end{align*}

Likewise, if all $s_i$'s are as claimed, then 
\begin{align*}
\finger(S) &
= \sum_{i=1}^{b} s_i(r) \alpha^i 
= \sum_{i=1}^{b} (\sum_{y=1}^{v} f_x(r,y) f_i(r,y) \big) \alpha^i
\end{align*}

Thus, if the $s_i$'s are as claimed, then these two fingerprints
should match. 
Moreover, by the Schwartz-Zippel lemma, and the fact that $\alpha$ and
$r$ are picked randomly by $\verifier$ and not known to $\helper$,
the fingerprints will not match with high probability if the $s_i$'s are \textit{not} as claimed, 
when the polynomials are evaluated
over a field of size polynomial in $(b+c)$. 

\eat{
the protocol is effectively computing 
a vector of $s_i(r)$ values (for the polynomial claimed to be
$\sum_{y=1}^v f_i(r,y) f_x(r,y)$) and ensuring that it is equal to the
vector of computed $\sum_{y=1}^v f_i(r,y) f_x(r,y)$ values. 
Each fingerprint is of a vector of $f_i(r,y)$ values. 
By linearity, we have that 
$\finger(f_i(r,y)) f_x(r,y) = \finger(f_i(r,y) f_x(r,y)$
and
$\sum_{y=1}^{v} \finger(f_i(r,y) f_x(r,y)) = \finger(\sum_{y=1}^{v}
f_i(r,y) f_x(r,y))$. 
}

To analyze the $\vcost$, we observe that
$\verifier$ can compute all fingerprints in
$O(v)$ space. 
As $\helper$ provides each polynomial $s_i(x)$ in turn, 
$\verifier$ can incrementally compute $\finger(S)$ and check that
this matches $\sum_{y=1}^{v} \finger(A,y) f_x(r,y)$. 
At the same time, $\verifier$ also computes $\sum_{i=1}^{b}
\sum_{x=1}^{h} s_i(x)$, as the value of $A\mathbf{x}$.
Note that if each $s_i$ is sent one after another, $\verifier$ can
forget each previous $s_i$ after the required fingerprints and
evaluations have been made; and if $h$ is larger than $v$, does not
even need to keep $s_i$ in memory, but can instead evaluate it term by
term in parallel for each value of $x$.  
Thus the total space needed by $\verifier$ is dominated by the $v$
fingerprints and check values. 

Setting $h=c^{\alpha}$ and $v=c^{1-\alpha}$, the total size of the
information sent by $\helper$ is dominated by the $b$ polynomials of
degree $h = c^{\alpha}$.

To prove the lower bound, we give a simple reduction of {\sc index} to
matrix-vector multiplication. Suppose we have an instance $(x , k )$
of {\sc index} where $x \in \{0, 1\}^{n^2}$, 
$k \in [n^2]$. Alice constructs an $n \times n$ matrix $A$ from $x$ alone, in which $A_{i, j}=1$ if 
$x_{f(i, j) =1}$, where $f$ is a 1-1 correspondence $[ n ] \times [ n ] \rightarrow [n^2]$, 
and $A_{i, j}=0$ otherwise. Assume $f(i, j) = k$. Bob then constructs a vector $\mathbf{x} \in \mathbb{R}^n$ such that 
$\mathbf{x}_{i} = 1$ and  all other entries of $\mathbf{x}$ are 0. Then the $j$'th entry of $A\mathbf{x}$ is 1
if and only if  $x_{f(i, j) =1}$, and therefore the value of $x_{f(i, j)}$ can be extracted from the
vector $A\mathbf{x}$. Therefore, if we had an $(h, v)$ protocol for verifying matrix-vector multiplication
given an $n \times n$ matrix $A$
(even for a stream in which all entries of $A$ come before all entries of $\mathbf{x}$), we would obtain 
a $(\sqrt{h}, \sqrt{v})$ protocol for {\sc index}. The lower bound for matrix-vector multiplication thus holds
by a lower bound for {\sc index} given in \cite[Theorem 3.1]{icalp09}.
\end{proof}

\begin{corollary}
\mbox{\!\!For $c\!\geq\!\!b$
there is a valid $(c^{1+\alpha}\!\!, c^{1-\alpha})$ protocol for the LP
streaming problem.}
\end{corollary}
\begin{proof} This follows by using the protocol of Theorem \ref{thm:matrixvector}
to verify $A\mathbf{x} \leq \mathbf{b}$ and $A^T \mathbf{y} \geq
\mathbf{c}$ within the protocol of Theorem \ref{LPtheorem}. 
The cost is $(bc^{\alpha} + cb^{\alpha}, c^{1-\alpha} +
b^{1-\alpha})$, so if  $c \geq b$, this is
dominated by $(c^{1+\alpha}, c^{1-\alpha})$ (symmetrically, if
$b > c$, the cost is $(b^{1+\alpha}, b^{1-\alpha})$). 
\end{proof}

Our protocol for linear programming relied on only two properties:
strong duality, and the ability to compute the value of a solution
$\mathbf{x}$ and check feasibility via matrix-vector multiplication.
Such properties also hold for more general convex optimization
problems, such as quadratic programming and a large class
of second-order cone programs.
Thus, similar results apply for these mathematical programs,
motivated by applications in which weak peripheral devices or sensors
perform error correction on signals. 
We defer full details from this presentation. 


Theorem \ref{thm:matrixvector} also implies the existence of 
protocols for graph problems where {\em both} $\hcost$ and $\vcost$ are sublinear in the
size of the input (for dense graphs). 
These include: 

\begin{itemize}

\item
An $(n^{1+\alpha}, n^{1-\alpha})$ protocol 
for verifying that $\lambda$ is an eigenvalue of the
adjacency matrix $A$ or the Laplacian $L$ of $G$: $\helper$ provides
the corresponding  eigenvector $x$, and $\verifier$ can use the
protocol of Theorem \ref{thm:matrixvector} 
to verify that $Ax=\lambda x$ or $Lx=\lambda x$. 

\item
An $(n^{1+\alpha}, n^{1-\alpha})$ protocol for the problem of
determining the effective resistance between designated nodes $s$ and
$t$ in $G$ where the edge weights are resistances.  
The problem reduces to solving an $n \times n$ system of 
linear equations
\cite{bollobastext}. 
\end{itemize}
%
\section{Simulating Non-Streaming Algorithms}
\label{sec:transcript}
Next, we give protocols by appealing to known
non-streaming algorithms for graph problems. 
At a high level, we can imagine the helper running an algorithm on the
graph, and presenting a ``transcript'' of operations carried out by
the algorithm as the proof to $\verifier$ that the final result is
correct. 
Equivalently, we can imagine that $\verifier$ runs the
algorithm, but since the data structures are large, 
they are stored by  $\helper$, who provides the contents of memory
needed for each step. 
There may be many choices of the algorithm to simulate and the
implementation details of the algorithm:  our aim is to
choose ones that result in smaller annotations.

To make this concrete, consider the case of requiring the graph to be
presented in a particular order, such as depth first order. 
Starting from a given node, the exploration retrieves nodes in 
order, based on the pattern of edges. 
Assuming an adjacency list representation, a natural implementation of
the search in the traditional model of computation 
maintains a stack of edges (representing the current path
being explored). 
Edges incident on the current node being explored are pushed, and pops
occur whenever all nodes connected to the current node have already
been visited. 
$\helper$ can allow $\verifier$ to recreate this exploration by providing
at each step the next node to push, or the new head of the stack when
a pop occurs, and so on. 
To ensure the correctness of the protocol, additional checking
information can be provided, such as pointers to 
the location in the stack when a node is visited that has already been
encountered. 

We provide protocols for breadth first search and depth first search
in Appendix \ref{app:dfs}, based on this idea of ``augmenting a
transcript'' of a traditional algorithm. 
However, while the resulting protocols are lightweight, it rapidly
becomes tedious to provide appropriate protocols for other computations 
based on this idea. 
Instead, we introduce a more general approach which argues that any 
(deterministic) algorithm to solve a given problem can be converted
into a protocol in our model. 
The running time of the algorithm in the RAM model becomes the size of
the proof in our setting. 

Our main technical tool
is the off-line memory checker of Blum et al. \cite{memcheck}, 
which we use to efficiently verify a sequence of accesses to a large memory. 
Consider a {\em memory transcript} of a sequence of read and write
operations to this memory (initialized to all zeros).
Such a transcript is {\em valid} if each read of address $i$ returns
the last value written to that address. 
The protocol of Blum et al. requires each read to be accompanied by
the timestamp of the last write to that address; and to treat each
operation (read or write) as 
a read of the old value followed by the write of a new value. 
Then to ensure validity of the transcript, it suffices to check that a fingerprint of all write operations
(augmented with timestamps) matches a fingerprint of all read
operations (using the provided timestamps), along with some simple
local checks on timestamps. 
Consequently, any valid (timestamp-augmented) transcript is accepted
by $\verifier$,
while any invalid transcript is rejected by $\verifier$ with high probability.


We use this memory checker to obtain the following general simulation result.

\begin{theorem} \label{sim} Suppose $P$ is a graph problem 
possessing a non-randomized algorithm $\algorithm$ 
in the random-access memory model
that, when given $G=(V,E)$ in adjacency list or adjacency matrix form, outputs 
$P(G)$ in time $t(m,n)$, where $m=|E|$ and $n=|V|$. 
Then there is an $(m + t(m,n), 1)$ protocol for $P$.
\end{theorem}

\begin{proof}[Proof sketch] 
$\helper$ first repeats (the non-zero locations of) a valid 
adjacency list or matrix representation $G$, as writes to the memory
(which is checked by $\verifier$);  $\verifier$  
uses fingerprints to ensure the edges included in the representation 
precisely correspond to those that appeared in the stream, and can 
use local checks to ensure the representation is otherwise valid. 
This requires $O(m)$ annotation and effectively initializes memory 
for the subsequent simulation. 
Thereafter, $\helper$ provides a valid 
augmented transcript $T'$ of the read and write operations performed 
by algorithm $\algorithm$; 
$\verifier$ rejects if $T'$ is invalid, or if  any  
read or write operation executed in $T'$ does not agree with the 
prescribed action of $\algorithm$. 
As only one read or write operation is 
performed by $\algorithm$ in each timestep, the length of $T'$ is $O(t(m,n))$, 
resulting in an $(m+t(m,n), 1)$ protocol for $P$. 
\end{proof}

\eat{
Note that this outline allows $\algorithm$ to be somewhat non-deterministic: it
can observe $G$ and then ``guess'' a solution. 
Provided $\algorithm$ then (deterministically) {\em checks} that the
solution satisfies $P$, then $\verifier$ can accept this solution. 
We obtain as corollaries of this simulation result tight protocols for
a variety of canonical graph problems.
}

Although Theorem \ref{sim} only allows the simulation of deterministic
algorithms, $\helper$ can non-deterministically ``guess" an optimal solution
$S$ and prove optimality by invoking Theorem \ref{sim} on a
(deterministic) algorithm that merely checks whether $S$ is optimal.
Unsurprisingly, it is often the case that the best-known algorithms
for verifying optimality are more efficient than those finding a
solution from scratch (see e.g. the MST protocol below); therein lies
much of the power of the simulation theorem.


%

\begin{theorem}
There is a valid $(m, 1)$ protocol to find
a minimum cost spanning tree; a valid $(m + n \log n, 1)$ protocol to verify single-source shortest paths;
and a valid $(n^3, 1)$ protocol to verify all-pairs shortest paths.
\end{theorem}

\begin{proof} 
We first prove the bound for MST. Given a spanning tree $T$, there
exists a linear-time algorithm $\algorithm$ for verifying that $T$ is minimum
(see e.g. \cite{king:97}). 
Let $\algorithm'$ be the linear-time algorithm that,
given $G$ and a subset of edges $T$ in adjacency matrix form, first
checks that $T$ is a spanning tree by ensuring $|T|=n-1$ and $T$ is
connected (by using e.g. breadth-first search), and then executes $\algorithm$
to ensure $T$ is minimum.  We obtain an $(m,1)$ protocol for MST by
having $\helper$ provide a minimum spanning tree $T$ and using Theorem
\ref{sim} to simulate algorithm $\algorithm'$.

The upper bound for single-source shortest path follows from Theorem
\ref{sim} and the fact that there exist implementations of Djikstra's
algorithm that run in time $m + n \log n$. The upper bound for
all-pairs shortest paths also follows from Theorem \ref{sim} and the
fact that the Floyd-Warshall algorithm runs in time $O(n^3)$.  
\end{proof}

We now provide near-matching lower bounds for all three problems.
\begin{theorem}
Any protocol for verifying single-source or all pairs shortest paths requires 
$\hcost \cdot \vcost = \Omega(n^2)$ bits. Additionally, if edge weights may be
specified incrementally, then an identical lower bound holds for MST.
\end{theorem}

\begin{proof} 
The lower bounds for single-source and all-pairs shortest paths are
inherited from shortest $s-t$ path (Corollary \ref{cor:lb}).

To prove the lower bound for MST, we present a straightforward
reduction from an instance of {\sc index}, $(x , k )$, where $x \in
\{0, 1\}^{n^2}$, $k \in [n^2]$.  Alice will construct a graph $G$,
with $V_G = [n]$, and $E_G = E_A$. Bob will then construct two graphs,
$G_1$ and $G_2$, with $E_{G_1} = E_A \cup E_{B_1}$ and $E_{G_2} = E_A
\cup E_{B_2}$.  If edge $(i, j)$ is in $E_A \cap E_{B_1}$, then we
interpret this to mean that the weight of edge $(i, j)$ in $E_{G_1}$
is the \textit{sum} of its weights in $E_A$ and $E_{B_1}$. 
Below, we will write $(i, j, w)$ to denote an edge 
between nodes $i$ and $j$ with weight $w$. 

Alice creates $E_A= \{(i, j, 1): x_{f(i, j) =1}\}$ from $x$ alone, where
$f$ is a bijection $[n] \times [n] \rightarrow [n^2]$. 
Bob creates $E_{B_1} = \{(u, v, 3) : f(u, v) \neq k\}$, and 
$E_{B_2} = E_{B_1} \cup \{(i, j, 1)\}$, where $(i, j)$ is the edge satisfying $f(i, j)=k$. Edge $(i, j)$, if it exists, 
is the lowest-weight edge in $E_{G_1}$, and hence $(i, j)$ is in 
any min-cost spanning tree of $G_1$ if and only if $x_k =1$. In contrast, $(i, j)$ is always in
the min-cost spanning tree of $G_2$. Therefore,
if $x_k =1$, then the minimum spanning tree of $G_2$ will be of higher cost than
that of $G_1$, because the weight of $(i, j)$ is $1$ in $E_{G_1}$ and $2$ in $E_{G_2}$. 
And if $x_k =0$, then the mininmum spanning tree of $G_2$ will be of lower cost than
that of $G_1$, because the weight of edge $(i, j)$ will be $\infty$ in $G_1$ and 1 in $G_2$.
Thus, by comparing the cost of the MSTs of $G_1$ and $G_2$, Bob can extract the value
of $x_{k}$. The lower bound now follows from the  hardness of {\sc index}  \cite[Theorem 3.1]{icalp09}.
\end{proof}

\para{Diameter.}
The diameter of $G$ can be verified via the all-pairs shortest path
protocol above, but the next protocol improves over the memory
checking approach. 


\begin{theorem} There is a valid $(n^2 \log n, 1)$ protocol for
 computing graph diameter. Further, any protocol for diameter 
 requires $\hcost \cdot \vcost = \Omega(n^2)$ bits.
\end{theorem}

\begin{proof}  \cite[Theorem 5.2]{icalp09} gives an $(n^2\log l, 1)$ protocol for verifying that $A^l=B$ for 
a matrix $A$ presented in a data stream and for any positive integer $l$. 
Note that if $A$ is the adjacency matrix of $G$; then $(I+A)^l_{ij} \neq 0$ if and only if 
there is a path of length at most $l$ from $i$ to $j$. Therefore, 
the diameter of $G$ is equal to the unique $l>0$ such that $(I+A)^l_{ij} \neq 0$ for all $(i, j)$,
while $(I+A)^{l-1}_{ij} = 0$ for some $(i, j)$. Our protocol requires $\helper$ to
send $l$ to $\verifier$, and then run the protocol of \cite[Theorem 5.2]{icalp09} twice
to verify that $l$ is as claimed. Since the diameter is at most $n-1$, 
this gives an $(n^2 \log n, 1)$ protocol. 

We prove the lower bound via a reduction from an instance of {\sc index},
$(x , k )$,
where 
$x \in \{0, 1\}^{n^2/4}$, $k \in [n^2/4]$. 
Alice creates a bipartite graph $G=(V, E)$ 
from $x$ alone: she 
includes edge $(i, j)$ in $E$ if and only if
$x_{f(i, j)}=1$, where $f$ 
is a bijection from edges to indices.
Bob then adds to $G$ two nodes $L$ and $R$, 
with edges from $L$ to each node in the left partite set,
edges from $R$ to each node in the right partite set, and an edge 
between $L$ and $R$. 
This ensures that the graph is connected, with diameter at most 3.
Finally, Bob appends a path of length 2 to node $i$, and a path of length 2 to node $j$, where $f(i, j)=k$.
If $x_{k}=0$, then the diameter is now 7, while if  $x_{k}=1$, the
diameter is 5. 
The lower bound follows from
the  hardness of {\sc index}  \cite[Theorem 3.1]{icalp09}
(this also shows that any protocol to approximate diameter better than
$\sqrt{1.4}$ requires $\hcost \cdot \vcost = \Omega(n^2)$ bits;
no effort has been made to optimize the  inapproximability constant).
\end{proof}

\section{Conclusion and Future Directions}
In this paper, we showed that a host of graph problems possess 
streaming protocols requiring only constant space and linear-sized 
annotations. 
For many applications of the annotation model, the priority 
is to minimize $\vcost$, and these protocols achieve this goal. 
However, 
these results are qualitatively different from those involving numerical 
streams in the earlier work \cite{icalp09}: for the canonical problems
of heavy hitters, frequency 
moments, and selection, it is trivial to achieve an $(m, 1)$ protocol 
by having $\helper$ replay the stream in sorted (``best'') order. 
The contribution 
of \cite{icalp09} is in presenting protocols obtaining optimal tradeoffs 
between $\hcost$ and $\vcost$ in which both quantities are sublinear 
in the size of the input. There are good reasons to seek these tradeoffs. 
For example, consider a verifier with access to a few MBs or GBs of
working memory.
If an $(m, 1)$ protocol requires only a few KBs of space, it would be 
desirable to use more of the available memory to significantly 
reduce the running time of the verification protocol.

In contrast to \cite{icalp09}, it is non-trivial to obtain $(m, 1)$
protocols for the  graph problems we consider, 
and we obtain tradeoffs involving sublinear values of $\hcost$ and 
$\vcost$ for some problems with an algebraic flavor
(e.g. matrix-vector multiplication, computing effective resistances,
and eigenvalues of the Laplacian). 
We thus leave as an open question whether it is possible 
to obtain such tradeoffs for a wider class of graph problems, and in
particular if the use of memory checking can be adapted to 
provide tradeoffs. 

A final open problem is to ensure that the work of $\helper$ is
scalable. 
In motivating settings such as Cloud computing environments, 
the data is very large, and $\helper$ may represent a {\em
  distributed} cluster of machines.  
It is a challenge to show that these protocols can be executed in a
model such as the MapReduce framework. 


\vspace{.05in}
\para{Acknowledgements.}
We thank Moni Naor for suggesting the use of memory checking. 
%

%
\bibliographystyle{IEEEtran}
\bibliography{biblio}

\appendix
\newpage

\section{Direct Protocols for Breadth First and Depth First Search}
\label{app:dfs}

In this appendix, we consider the problem of ``verifying" a breadth first or depth first search. 
Our goal is to force $\helper$ to provide to 
$\verifier$ the edges of $G$ in the order they would be visited by a
 non-streaming BFS or DFS algorithm, despite limiting $\verifier$ to 
 logarithmic space. We can accomplish this with linear annotation simply by invoking 
 Theorem \ref{sim} on a BFS or DFS algorithm; however we now give ``direct" protocols for this problem.
Our protocols improve over Theorem \ref{sim} by constant factors, because we avoid having to 
replay the entire contents of memory at the beginning and end of the protocol. 
As an immediate corollary of our BFS protocol, 
 we obtain an $(m, 1)$ protocol for bipartiteness. 

To make precise the notion of ``verifying a BFS'', 
we define the concept of a BFS transcript, using the concept of a
label-augmented list of edges defined in Definition \ref{labeldef}. 

\begin{definition}  \label{bfstranscript}  
Given a connected undirected graph $G$, a BFS 
transcript $T$ rooted at $s$ is a label-augmented list of edges $E'$, with 
each label $l(e,u)$ claimed to be the distance from $s$ to $u$. 
Let $d_e:= \min(l(e,u), l(e,v))$. 
A BFS transcript is {\em valid} if: 
(a) Edges are presented in increasing order of $d_e$; 
(b) $E'$ is a valid list of label-augmented edges; and 
(c) For all $u$, $l(u)$ is the (hop) distance from $s$ to $u$.
 \end{definition}
  
For clarity, we make explicit the labels of each edge; a more concise
protocol would group edges in order of
$l(e,u) + l(e,v)$. 
However, this does not alter the asymptotic cost. 

It is easy to see that any valid BFS transcript corresponds to the order 
edges are visited in an actual BFS of $G$. We therefore say a protocol 
verifies a BFS if it accepts  any valid BFS transcript (possibly with 
additional annotation interleaved) and rejects any invalid BFS transcript 
with high probability. With this in mind, we now define augmented BFS transcripts.

Any valid BFS transcript corresponds to a possible order of edges
being visited in an actual BFS of $G$. 
So we seek a protocol that accepts any valid BFS transcript
(possibly with additional annotation interleaved), and rejects an
invalid transcript (whp).  This leads us to define an augmented BFS
transcript. 

Let $T$ be a BFS transcript. 
We partition the edges by their distance, so that 
 $E_l=\{e: d_e = l\}$.
Likewise, we define 
$V_l=\{u | \exists e=(u, v) \mbox{ s.t. } l(e,u)=l\}$. 
In a valid BFS transcript, the sets $V_l$  partition  $V$
since $l(e,u)$ gives the same distance from $s$ to $u$
every time is is listed; an invalid transcript may not have this property. 

\begin{definition} \label{augmentedbfs} 
An  augmented 
BFS transcript $T'$ is a BFS transcript $T$ 
where additionally before each $E_l$, 
$V_{l+1}$ is presented, 
along with degrees $\deg_l(v)$, defined as 
$\deg_l(v) = |\{ u | (u,v) \in E, u \in V_{l}, v \in V_{l+1}\}|$. 

The augmented BFS transcript
is valid if the underlying BFS transcript is valid 
and all claimed degrees are truthful. 
\end{definition}
\begin{theorem} \label{BFStheorem}
\mbox{
There is a valid $(m,1)$ protocol to accept a valid augmented BFS
transcript.} 
\end{theorem}
\begin{proof} 
To ensure Definition \ref{bfstranscript} (a),
$\verifier$ rejects if the edges are not presented in increasing order
of $d_e$, which is trivial to check  
since $l(e,u)$ and $l(e,v)$ are presented for every edge. 
$\verifier$ ensures Definition \ref{bfstranscript} (b) 
using the ``Consistent Labels'' protocol of Lemma \ref{labellemma}.

To verify the augmented transcript, 
fingerprints ensure multiset $\{ v \in V_{l+1} : \deg_l(v)\}$ matches 
the multiset of edges $e=(u,v)$ with $l(e,u)=l, l(e,v)=l+1$.
$\verifier$ (re)uses the same working memory to store these fingerprints
for each level $l$ in turn. 
To ensure Definition \ref{bfstranscript} (c),
$\verifier$ uses the augmented annotation at each edge-level to
reject if a node at level $l+1$ is not incident  
to any nodes at level $l$, or if an edge is ever presented 
with $|l(e,u) - l(e,v)| > 1$. 
An inductive argument ensures  
that $l(e,u)$ is the distance from $s$ to $u$ for all $u$ and all $e$
incident to $u$ 
(the base case is just $l(s)=0$). 
For the $\vcost$,
note that any valid augmented BFS transcript has length $O(n+m) = O(m)$.
\end{proof}
From this theorem, we obtain an $(m, 1)$ protocol 
for bipartiteness:
$G$ is bipartite if and only if
there is no edge $(u,v)$  with $l(e,u) = l(e,v)$, which is easily
checked given a valid augmented BFS transcript.
Further, any online 
protocol for bipartiteness requires $\hcost\cdot \vcost = \Omega(n)$ bits, 
even when $m=O(n)$;
this lower bound follows by reducing  bipartiteness to
{\sc index} \cite{feigenbaum2}, which has the same lower bound 
\cite[Theorem 3.1]{icalp09}. 

Our approach for Depth First Search is similar, but more involved. 
Analogously to BFS, we now define DFS transcripts.

\begin{definition} \label{dfstranscript}
Given a connected undirected graph $G$, a DFS transcript consists of
$m+2n$ rows, where each row is of the form
$\edge(u,v)$, $\push(u)$, or $\pop(u)$.
Each symbol $\push(u)$ (resp. $\pop(u)$)
represents a simulated push (pop) of node $u$ on a
(simulated) stack.
Reading the transcript in order,
the most recently pushed vertex that has not yet been
popped is termed the ``top'' node.
A DFS transcript is \textit{valid} for $G=(V,E)$
if \\
(a) Every $v\in V$ is pushed exactly once, immediately after it first appears in an edge; \\
(b) Every $e\in E$ is presented exactly once, and is incident to the top node; \\
(c) Every $v\in V$ is popped exactly once, when it is the top and all edges
incident to it have already appeared in the transcript.
\end{definition}

The order edges are visited in an actual DFS of $G$ always
corresponds to the ordering of edges in a valid DFS transcript,
and each simulated push and pop operation corresponds to
actual push and pop operations in the same DFS.
Therefore, just as in the BFS case,
we seek a protocol that accepts any valid
DFS transcript
(possibly with additional annotation on each event) and
rejects invalid DFS transcripts whp.
We write
$t(\event,u)=t$ to denote that the $t$'th row of the transcript is
$\event(u)$.
Thus, the top node at time $t$ is
$top(t) = \arg\max_{u} (t(\push,u) | t(\push,u) \leq t < t(\pop,u))$, and
we say the stack height at time $t$ is
$height(t) = |\{ u | t(\push,u) \leq t < t(\pop,u)\}|$.

\eat{
To this end,
we say a list of edges $E'$ is timestamp-augmented if $E'$ is label-augmented
and moreover $\helper$ may annotate any edge $(u, v)$ with symbols of the form
$\event(u)$ or $\event(v)$, where $\event \in \{\push, \pop\}$;
$t(\event,u)=t$ denotes that the $t$'th row of the transcript is
$\event(u)$.
A list of timestamp-augmented edges $E'$ is valid if
for each node $u$, $\event(u)$ occurs exactly once, say at time $t(\event,u)$,
and for all $u$, $l(u)=t(\event,u)$.
}

\begin{definition} \label{augmenteddfs}
An augmented DFS transcript $T'$ is a DFS transcript where in addition:\\
(a) before the edges, triples $(u, level(u), ntop(u))$ are listed for
each node in order; \\
(b) every edge $e=(u,v)$ is annotated with
$(t(\push,u), t(\pop,u), t(\push,v), t(\pop,v))$.\\
(c) each $\pop(u)$ with $t(\pop,u)=t$
is annotated with $(v, t(\push,v), t(\pop,v))$ where $v=top(t+1)$
\end{definition}

\eat{
\begin{lemma} \label{timestamplemma}
There exists a protocol with $\vcost = \log m$ that verifies
that each occurrence of $u$ is labeled with the same value of
$t(\push,u)$ (respectively, $t(\push,u)$), and each
$\push(u)$ (resp. $t(\pop,u)$) occurs exactly once.
\end{lemma}

\begin{proof}
This is easily checked to see that only
It remains  to ensure that for all $u$, $\push(u)$
occurs exactly once, at time $t(\push,u)$.
This is also checked by fingerprinting
To this end, let $S = \{(u, t(\push,u))\}$ and
$S' = \{(u, t)| t(\event,u)=t \}$.
$\verifier$ first computes $\finger(S)$ from the annotation preceding
the list of edges.
Then, $\verifier$ computes $\finger(S')$ while observing $E'$.
$\verifier$ accepts if $\finger(S)=\finger(S')$, else correctly rejects.
\end{proof}}

\begin{definition} \label{validaugmenteddfs}
We say an augmented DFS transcript is valid if \\
(a) the underlying DFS transcript is valid; \\
(b) for all $u$, $level(u) = height(t(\push,u))$
and $ntop(u)$ is $|\{ t | top(t) = u\}|$. \\
(c) all $t(\push,u)$ and $t(\pop,u)$ annotations are consistent; \\
(d) the triple after each pop event correctly identifies the new top.
 \end{definition}
\begin{theorem}
\mbox{There is a valid $(m,1)$ protocol 
to accept a valid augmented DFS transcript.}
\end{theorem}

\begin{proof}
To ensure Definition \ref{validaugmenteddfs}~(b),
$\verifier$ initially computes a fingerprint $f_1$ of the multiset
$\{(u, level(u)): ntop(u)\}$.
Throughout the protocol, $\verifier$ also
tracks $height(t)$, and maintains a fingerprint $f_2$ of
all $(top(t), height(t))$ pairs observed while reading the augmented
DFS transcript.
If property (b) is satisfied,
$f_1$ and $f_2$ will match, otherwise they will differ with high
probability.
This is essentially the ``Consistent Labels'' protocol applied to this
setting.

For Definition \ref{validaugmenteddfs}~(c), we use
the ``Consistent Labels'' protocol of
Lemma \ref{labellemma} to ensure that the claimed $\push$ times are
consistent for all nodes. 
\eat{JT comment: I think the rest of this
paragraph is unnecessary (except for the last sentence). Firstly, \ref{validaugmenteddfs}~(c) only
requires consistency, and this is guaranteed by the Consistent Labels protocol.
Second, if the push timestamps are
all consistent, then by definition of consistency isn't there only one push
time claimed for each node? Another point: the Consistent Labels protocol
is only stated for a label-augmented list of edges, but we really need it to
work here for a label-augmented list of edges with events thrown in. It's
conceptually an easy fix: just view the events as edges in which a single
node appears. Should we just brush this under the rug?
GRC response: we still need to explicit verify defn 6a, 6c: that each
node is pushed exactly once in the original transcript.  That's what's
happening here.  Am happy to gloss over the remaining issues.
}
This protocol presents a list of nodes in sorted order accompanied
by their $t(\push,u)$ values.
This list is also used to check to see that exactly one $\push$ event 
is claimed for each node: a fingerprint of the list can then be compared to one
generated directly from the transcript.
The case for $\pop$ is identical.

\eat{
[GC comment: I don't know what this paragraph means]

We may ensure that Property (c) of Definition \ref{validaugmenteddfs}
is satisfied
by running a simple extension the timestamp protocol of Lemma
\ref{timestamplemma} with respect
to push and pop events (the extension accounts for the fact that in
addition to timestamping edges, we must
also assign timestamps to $v$ whenever $\helper$ claims $v$ is the new
top).
}

For Definition \ref{validaugmenteddfs}~(d),
$\verifier$ can reject if the triple of Definition \ref{augmenteddfs}~(c)
has $t(\push,v) > t$ or $t(\pop,v) < t$.
Observe that the stack must have the same height every time
$v$ is claimed to be the top, else the fingerprints $f_1$ and $f_2$
will differ (whp).
Thus we can assume that $v$ is always reported as top at the same
height, which is $h=height(t(\push,v))$.
We argue
that at any timestep $t$ there can be only one node $v$
at height $h$ such that $t(\push,v) \leq t < t(\pop,v)$.
If not,
then let $u$ and $v$ be a ``witness'' pair of
nodes both claimed to be at height $h$
such that $t(\push,v) < t(\push,u) < t(\pop,v)$:
if there are many such nodes, then pick $v$ with the
smallest value of $t(\push,v)$, and $u$ likewise with
the smallest $t(\push,u)$ such that the condition is met.
For this to happen, there must have been a $\pop$ event to bring the
height below $h$;
either this is to $v$ itself, contradicting the assumptions, or to
some other node, $w$.
Assuming that $t(\push,w) < t(\pop,w)$ (else our other checks identify
this error), then $v$ and $w$ form a witness pair
with $t(\push,w) < t(\pop,w) < t(\push,u)$,
contradicting the claim that $u$ and $v$ formed the
earliest witness pair.

\eat{ To ensure Property (d) of Definition \ref{validaugmenteddfs}
 holds, we have $\verifier$ reject if $\helper$ ever claims a node
 $v$ is the new top, but $v$'s push timestamp is in the future or
 $v$'s pop timestamp is in the past.  To see that this suffices for
 ensuring Property (d), note that if a pop occurs bringing the height
 of the stack to $k$, and $\helper$ states that $v$ is the new top
 but $v$ was not pushed when the stack was at height $k$, then
 fingerprints $f_1$ and $f_2$ will differ with high probability,
 causing $V$ to reject. Therefore, a simple inductive argument shows
 that at any timestep $t$, there is only one node who has been pushed
 at height $k$ and has not yet been popped (and this is the true
 top); we conclude that $\helper$ must provide the true top after a
 pop at time $t$, or cause $\verifier$ to reject with high
 probability.
}

Next, we ensure that the underlying DFS transcript is valid.
Definition \ref{dfstranscript}~(a) is easy to check, given
Definition \ref{validaugmenteddfs}~(c), by having $\verifier$ reject
if the $t$'th row of the transcript is
 $\edge(u,v)$ and is annotated with $t(\push,v)>t+1$.
 Definition \ref{dfstranscript}~(b) is also easy to check given
 Definition \ref{validaugmenteddfs}~(d): $\verifier$ tracks $top(t)$
 implicit from $\push$ operations or from annotations
 on $\pop$ events.
This leaves Definition \ref{dfstranscript}~(c) i.e. to ensure that
nodes are popped at the correct times.
To accomplish this, it suffices to have $\verifier$ reject if
the $t$'th row of the transcript is
$\edge(u,v)$ and is annotated with $t(\pop,u)<t$
or if the $t$'th row is
$\pop(u)$ and $top(t)\neq u$.
These rules ensure $\pop(v)$ cannot occur ``too early'', i.e.
before all edges incident to $v$ have been presented:
when such an edge is presented in row $t$,
$\verifier$ will reject due to $t(\pop,v)<t$.
Likewise, $v$ cannot be popped ``too late'',
since if $top(t)=v$ and there are no more edges incident on $v$,
the next edge will not be incident on $top(t)$ (required by
Definition \ref{dfstranscript}~(b)) unless $v$ is popped first.
\end{proof}

\end{document}